\documentclass[11pt]{article}
\usepackage{amsfonts}
\usepackage{amsmath}
\usepackage{amssymb}
\newtheorem{theorem}{Theorem}
\newtheorem{lemma}[theorem]{Lemma}
\newtheorem{proposition}[theorem]{Proposition}
\newtheorem{corollary}[theorem]{Corollary}

\newtheorem{remark}{Remark}

\newcommand{\F}{\mathbb{F}}

\newenvironment{proof}{\noindent {\em Proof.}}{\hspace*{\fill} $\Box $\newline}

\title{On the Structure of the Binary LCD Codes having an Automorphism of Odd Prime
Order}
\author{Stefka Bouyuklieva\\ St. Cyril and St. Methodius University of Veliko Tarnovo, Bulgaria\\ Javier de la Cruz \\Universidad del Norte, Barranquilla, Colombia
}
\date{}
\begin{document}
\maketitle

\begin{abstract} The aim of this work is to study the structure and properties of the binary LCD codes having an automorphisms of odd prime order and to present a method for their construction.
\end{abstract}

\section{Introduction}
\label{intro}

We study the structure of the binary linear codes having an automorphism of an odd prime order $p$, and provide some necessary and sufficient conditions for such code to be a \textit{linear complementary dual} (LCD) code, i.e. the intersection of the code and its orthogonal complement should consists only of the zero vector $\mathbf{0}$.

Let $\mathbb{F}_q$ be a finite field with $q$ elements and $\mathbb{F}_q^n$ be the $n$-dimensional vector space over $\mathbb{F}_q$. The (Hamming) \textit{distance} $d(x,y)$ between two vectors $x,y\in\mathbb{F}_q^n$ is the number of coordinate positions in which they differ. The (Hamming) \textit{weight} wt$(x)$ of a vector $x\in\mathbb{F}_q^n$ is the number of its nonzero coordinates.
  A linear $[n,k,d]$ code $C$ is a $k$-dimensional subspace of the
vector space $\mathbb{F}_q^n$, and $d$ is the
smallest weight among all non-zero codewords of $C$, called the \textit{minimum weight} (or minimum distance) of the code. A matrix whose rows form a basis of $C$
is called a \textit{generator matrix} of this code. The \textit{weight enumerator}
of a code $C$ is given by the polynomial $W_C(y)=\sum_{i=0}^n A_iy^i$ where
$A_i$ is the number of codewords of weight $i$ in $C$.
In the binary case, i.e. $q=2$, two codes are equivalent if one can be obtained from the other by a permutation of coordinates.
The set of all permutations that preserve a linear code $C$ of length $n$, forms a subgroup of the symmetric group $\mathrm{Sym}(n)$, denoted by $\mathrm{Aut}(C)$ and called the automorphism group of the code.

Let $(u,v):
\mathbb{F}_q^n\times\mathbb{F}_q^n\to \mathbb{F}_q$ be an inner
product in the linear space $\mathbb{F}_q^n$. The orthogonal complement of $C$ according to the defined inner product is called the \textit{dual code} of $C$ and denoted by $C^{\perp}$, i.e. $C^{\perp}=\{u \in \mathbb{F}_q^n : (u,v)=0$ for all $v \in C\}$. Obviously, $C^{\perp}$ is a linear $[n,n-k]$ code. If $C \subseteq
C^{\perp}$, $C$ is termed \textit{self-orthogonal} and if  $C = C^{\perp}$,
$C$ is \textit{self-dual}. 

Although the concept of a zero-dimensional intersection of a  vector space and its orthogonal complement is very natural in linear algebra, the   LCD codes over finite fields were introduced by Massey in his 1992 paper \cite{Massey}.
Recently, much work has been done on the study, constructions and classifications of LCD codes with or without additional restrictions over different finite fields and rings
 (see, for example,  \cite{Carlet_Pellikaan,Dougherty-Kim-Sole-LCD,Kim-LCD,Harada-LCD}). It seems, that LCD and self-dual codes are completely different classes of codes, but sometimes the same approaches can be very useful for construction of codes from both classes. The method of constructing and classifying binary self-dual codes having an automorphism of odd prime order is a very powerful tool in the theory of this class of codes \cite{Huff48,Yorgov56}. Our research is devoted to the following problem: What is the structure of the binary LCD codes having an automorphism of odd prime order $p$, and how to use this structure for construction and classification? Therefore, we first study the structure of binary linear codes invariant under an automorphism of order $p$, the corresponding structure of their dual codes, and the relationship between these structures. The work on this problem began last year and the first results were presented in \cite{ACCT2020}, but the authors there consider only primes $p$, for which 2 is a multiplicative root.

This paper is organized as follows. In Section \ref{sect:linear} we study the structure of binary linear codes having an automorphism of odd prime order. Section \ref{sect:LCD} is devoted to the linear complementary dual (LCD) codes having such automorphisms. Finally, in Section \ref{sect:aut7}, we give some construction and classification results for optimal LCD codes having an automorphism of order $p$ for $p=7$, 11, and 17.



%

\section{Binary linear codes with an automorphism of odd prime order}
\label{sect:linear}

Let $C$ be a binary linear code of length $n$ and $\sigma$ be an
automorphism of $C$ of odd prime order $p$ with $c$ independent
$p$-cycles. Without loss of generality we can assume that
\begin{equation}\label{perm}
\sigma =\Omega_1\dots\Omega_c\Omega_{c+1}\dots\Omega_{c+f}
\end{equation}
where $\Omega_{i}=((i-1)p+1,\dots,ip), i=1,\dots,c$, are the
cycles of length $p$, and $\Omega_{c+i}=(cp+i), i=1,\dots,f$, are
the fixed points. Obviously, $cp+f=n$.

Let $F_{\sigma}(C)=\{ v\in C:v\sigma=v\}$ and $E_{\sigma}(C)$=$\{
v\in C:wt(v\vert \Omega_i)\equiv 0\pmod{2},i=1,\dots,c+f\} $,
 where $v\vert\Omega_i$ is the restriction of $v$ on $\Omega_i$.
    The following lemma is a modification of the general decomposition lemma (Lemma 2) given in \cite{Huff48}. As Huffman has mentioned there, part of its statement and proof
are a special case of Maschke’s theorem.

\begin{lemma}{\rm \cite{Huff48}}\label{lemma:structure}
 The code C is a direct sum of the subcodes $F_{\sigma}(C)$ and $E_{\sigma}(C)$, which are mutually orthogonal, $F_{\sigma}(C)\perp E_{\sigma}(C)$.
\end{lemma}

%
%
%
%


The subcodes from Lemma \ref{lemma:structure} are the main building blocks of the codes we consider. Later in this section we study them in more details.

As $v\in F_{\sigma}(C)$ if and only if $v\in C$ and $v$ is constant on
    each cycle, the projection map $\pi :F_{\sigma}(C)\to \F_{2}^{c+f}$, defined by
$ (v\pi)_i =v_j $ for some $ j\in\Omega_i, i=1,2,\ldots,c+f$, $v\in F_{\sigma}(C)$, is a monomorphism. It turns out that the binary codes $F_{\sigma}(C)$ and $C_{\pi}=\pi(F_{\sigma}(C))$ are isomorphic as linear spaces over $\F_2$ and therefore they have the same dimenstion $k_{\pi}=\dim C_{\pi}=\dim F_{\sigma}(C)$. Moreover,
$$x\cdot y=p\sum_{i=1}^c x_iy_i+\sum_{i=c+1}^{c+f}x_iy_i= \sum_{i=1}^c x_iy_i+\sum_{i=c+1}^{c+f}x_iy_i=\pi(x)\cdot\pi(y)$$
for any two vectors $x,y\in F_{\sigma}(C)$, $x=(\underbrace{x_1,\ldots,x_1}_p,\ldots,\underbrace{x_c,\ldots,x_c}_p,x_{c+1},\ldots,x_{c+f})$,\\
$y=(\underbrace{y_1,\ldots,y_1}_p,\ldots,\underbrace{y_c,\ldots,y_c}_p,y_{c+1},\ldots,y_{c+f})$. By $u\cdot v=\sum_{i=1}^n u_iv_i\in\F_2$ we denote the Euclidean inner product of the binary vectors $u,v\in\F_2^n$. The following lemma is very helpful for some of the proofs in Section \ref{sect:LCD}.

\begin{lemma}\label{lemma:Cpi}
If $C$ is a binary linear code invariant under the permutation $\sigma$, given in \eqref{perm}, then $\pi(F_\sigma(C)\cap F_\sigma(C)^\perp)=C_{\pi}\cap C_{\pi}^{\perp}$.
\end{lemma}

\begin{proof}
If $x\in F_\sigma(C)\cap F_\sigma(C)^\perp$ then $\pi(x)\cdot\pi(y)=x\cdot y=0$ for any vector $y\in F_{\sigma}(C)$, so $\pi(x)\in C_{\pi}\cap C_{\pi}^{\perp}$. Conversely, let $\pi(x)\in C_{\pi}\cap C_{\pi}^{\perp}$. Then $x\cdot y=\pi(x)\cdot\pi(y)=0$ for all $y\in F_{\sigma}(C)$ and therefore $x\in F_\sigma(C)\cap F_\sigma(C)^\perp$.
\end{proof}

Denote by $E_{\sigma}(C)^{*}$ the code $E_{\sigma}(C)$ with the last $f$ coordinates deleted.
For $v\in E_{\sigma}(C)^{*}$  we
identify $v\vert\Omega_{i}=(v_{0},v_{1},\cdots ,v_{p-1})$ with the
polynomial $v_{0}+v_{1}x+\cdots +v_{p-1}x^{p-1}$ from $\mathcal{P}$,
 where $\mathcal{P}$ is the set of even-weight polynomials in
$\F_{2}[x]/(x^p-1)$. Thus we obtain the map $\varphi:E_{\sigma}(C)^{*}\rightarrow \mathcal{P}^c$. Obviously, $C_{\varphi}=\varphi (E_{\sigma}(C)^*)$ is a $\mathcal{P}$-module, and if $\mathcal{P}$ is a field then $C_{\varphi}$ is a linear code.
On $\mathcal{P}^c$, we use the Hermitian inner product, defined in
\cite{LingSoleI}, namely
\begin{equation}\label{eq:product}
\langle u,v\rangle=\sum_{j=1}^c u_j\overline{v_j},
\end{equation}
where $u=(u_1,\dots,u_c)$, $v=(v_1,v_2,\dots,v_c)\in \mathcal{P}^c$, and $\overline{v_j}=v_j(x^{-1})=v_j(x^{p-1})$.

The code $E_{\sigma}(C)^*$ is equivalent to a quasi-cyclic binary code of length $cp$ and index $c$, and we can use the decomposition given in \cite{LingSoleI}.
Let  $\mathcal{R}_p=\F_2[x]\slash (x^p-1)$ and
\begin{equation}\label{eq:h}
x^p-1=(x-1)g_1(x)\cdots g_s(x)h_1(x)h_1^*(x)\cdots h_t(x)h_t^*(x),
\end{equation}
where $f^*(x)$ is the reciprocal polynomial of $f(x)$, $g_i^*(x)=g_i(x)$, $i=0,1,\ldots,s$, $g_0(x)=x-1$, and $h_j(x)\neq h_j^*(x)$ for $j=1,\ldots,t$.
If $G_i=\F_2[x]\slash (g_i(x))$, $i=0,\ldots,s$, $H_j=\F_2[x]\slash (h_j(x))$ and $H_j^*=\F_2[x]\slash (h_j^*(x))$, $j=1,\ldots,t$, then
\begin{equation}
\mathcal{R}_p=G_0\oplus G_1\oplus\cdots\oplus G_s\oplus H_1\oplus H_1^*\oplus\cdots\oplus H_t\oplus H_t^*.
\end{equation}

If we denote by $\widehat{f(x)}$ the polynomial $\frac{x^p-1}{f(x)}$ for any divisor $f(x)$ of $x^p-1$, then $G_0,\ldots,G_s$, $H_1,\ldots,H_t$, $H_1^*,\ldots,H_t^*$ are isomorphic to the ideals $\langle \widehat{g_i(x)}\rangle$, $i=0,\ldots,s$, $\langle \widehat{h_j(x)}\rangle$ and $\langle \widehat{h_j^*(x)}\rangle$, $j=1,\ldots,s$, respectively. 
Let $I_j=G_j$ for $j=0,\ldots,s$, $I_{s+j}=H_j$ and $I_{s+t+j}=H_j^*$ for $j=1,\ldots,t$. We denote the generating idempotents of $I_j$ by $e_j(x)$.
The following theorem is a modification of \cite[Theorem 4.3.8]{HP}.


\begin{theorem}\label{thm:I_j} The following holds in $\mathcal{R}_p$:
\begin{itemize}
\item[(i)] The ideals $I_j$ for $0\le j\le s+2t$ are all the minimal ideals of $\mathcal{R}_p$;
\item[(ii)] $e_i(x)e_j(x)=0$ for $i\neq j$;
\item[(iii)] $e_0(x)+\cdots+e_{s+2t}(x)=1$ in $\mathcal{R}_p$;
\item[(iv)] $0$ and $e_j(x)$ are the only idempotents in $I_j$, $0\leq j\leq s+2t$.
\end{itemize}
\end{theorem}

%

We also use another slightly different transcription of the factorization \eqref{eq:product} of $x^p-1$ into irreducible factors over $\F_2$, namely
$x^p-1=\prod_m M_{\alpha^m}(x)$, where $m$ runs through a
set of representatives of the 2-cyclotomic cosets modulo $p$, and $\alpha$ is a primitive $p$-th root of unity \cite[Theorem 4.1.1]{HP}. Furthermore, the size of each 2-cyclotomic coset is a divisor of $\mathrm{ord}_p(2)$ \cite[Theorem 4.1.1]{HP}. Since $p$ is an odd prime, all 2-cyclotomic cosets modulo $p$ have the same size. Hence all irreducible factors of $\frac{x^p-1}{x-1}$ have the same degree, namely $\mathrm{ord}_p(2)$. It turns out that $s+2t=\frac{p-1}{\mathrm{ord}_p(2)}$ and $\deg g_1=\cdots=\deg g_s=\deg h_1=\deg h_1^*=\cdots=\deg h_t=\deg h_t^*=\mathrm{ord}_p(2)$. Moreover, $I_j$ are fields with $2^{\mathrm{ord}_p(2)}$ elements, $j=1,\ldots,s+2t$.

We have to mention here that $e_j(x)=\sum_{j\in J}\sum_{i\in C_j} x^{i}$, where $J$ is some subset of representatives of 2-cyclotomic cosets modulo $p$, and $e_j(x^{-1})=e_i(x)$ for some $i$ (see \cite[Corollary 4.3.15]{HP}). If $e_i(x)$ is the nonzero idempotent of $G_i$ then $e_i(x)=u(x) \widehat{g_i}(x)$ and so $e_i(x^{-1})=u(x^{-1})\widehat{g_i}(x^{-1})=u(x^{-1})x^{-m}\widehat{g_i}(x)\in G_i$, which gives us that $e_i(x^{-1})=e_i(x)$, $i=1,\ldots,s$. If $e_{s+j}(x)$ is the nonzero idempotent of $H_j$ then $e_{s+j}(x)=u(x) \widehat{h_j}(x)$ and so $e_{s+j}(x^{-1})=u(x^{-1})\widehat{h_j}(x^{-1})=u(x^{-1})x^{-m}\widehat{h_j^*}(x)\in H_j^*$, which gives us that $e_{s+j}(x^{-1})=e_j^*(x)=e_{s+t+j}(x)$, $j=1,\ldots,t$.

Since $g_0(x)=x-1$ then $I_0=G_0\cong\F_2$, $\widehat{g_0}(x)=1+x+\cdots+x^{p-1}=e_0(x)=g_1(x)\cdots g_s(x)h_1(x)h_1^*(x)\cdots h_t(x)h_t^*(x)$ and
$$
\mathcal{P}=G_1\oplus\cdots\oplus G_s\oplus H_1\oplus H_1^*\oplus\cdots\oplus H_t\oplus H_t^*.
$$
Let $M_j=\{u=(u_1,\ldots,u_c)\in C_{\varphi} \vert \ u_i\in I_j, i=1,\ldots,c\}$, $j=1,\ldots,r$, where $r=s+2t$. Since $I_j$ is a field, $M_j$ is a linear code over $I_j$ of length $c$ and dimension $k_j\ge 0$. Then the following lemma holds.

\begin{lemma}\label{lem:Mj}
If \eqref{eq:h} is the decomposition of the polynomial $x^{p}-1$ into irreducible factors over $\F_2$ then
\begin{itemize}
\item[(i)] $C_\varphi=M_1\oplus M_2\oplus\cdots\oplus M_r$;
\item[(ii)] $k=k_{\pi}+m(k_1+\cdots+k_r)$, where $m=\mathrm{ord}_p(2)$.
\end{itemize}
\end{lemma}

\begin{proof} 
We can use the decomposition of a quasi-cyclic code, given in \cite{LingSoleI} but we prefer to follow the proof of \cite[Lemma 3]{Yorus}. Denote by $M$ the module $C_{\varphi}$. Then
$$M=Me(x)=Me_1(x)+\cdots+Me_r(x)=M_1+\cdots+M_r.$$
If $u=(u_1,\ldots,u_c)\in M_j\cap\sum_{i\neq j}M_i$ then $u_1,\ldots,u_c\in I_j\cap\sum_{i\neq j}I_i=\{ 0\}$. Hence $u=\mathbf{0}$ and $M$ is a direct sum of the modules $M_1,\ldots,M_r$. This gives us that
$$\dim_{\F_2} M=\sum_{j=1}^r\dim_{\F_2} M_j=\sum_{j=1}^r\dim_{I_j} M_j\dim_{\F_2} I_j=m\sum_{j=1}^r\dim_{I_j} M_j.$$
Since $\dim_{\F_2} M=\dim E_{\sigma}(C)$, then $\dim C=\dim F_{\sigma}(C)+\dim E_{\sigma}(C)=k_{\pi}+m\sum_{j=1}^rk_j$.
\end{proof}


Since $\mathrm{Aut}(C)=\mathrm{Aut}(C^\perp)$, the automorphism $\sigma\in\mathrm{Aut}(C)$ is also an automorphism of the dual code $C^\perp$. It turns out that $C^\perp=F_{\sigma}(C^\perp)\oplus E_{\sigma}(C^\perp)$. Denote $\pi(F_{\sigma}(C^\perp))$ by $\widehat{C_\pi}$, and $\varphi(E_{\sigma}((C^\perp)^*))$ by $\widehat{C_\varphi}$. The following theorem presents the relationship between the considered subcodes of $C$ and $C^{\perp}$.

\begin{theorem}\label{thm:C-Cperp}
 If $C$ is a binary linear code invariant under the permutation $\sigma$, given in \eqref{perm}, then $C_{\pi}^{\perp} =\widehat{C_\pi}$ and $C_{\varphi}^{\perp} = \widehat{C_\varphi}$.
\end{theorem}

\begin{proof}
Take $u\in C_\pi$ and $v\in \widehat{C_\pi}$. Then $\pi^{-1}(u)\in C$ and $\pi^{-1}(v)\in C^\perp$, therefore
$u\cdot v=\pi^{-1}(u)\cdot\pi^{-1}(v)=0$.
It follows that $v\in C^{\perp}_\pi$ and so $\widehat{C_\pi}\subseteq C_{\pi}^{\perp}$.

If $w\in C_{\pi}^{\perp}$ then
$0=u\cdot w=\pi^{-1}(u)\cdot\pi^{-1}(w)$.
Hence $\pi^{-1}(w)\perp F_{\sigma}(C)$. If $y\in E_{\sigma}(C)$ then
$$y\cdot\pi^{-1}(w)=\sum_{i=1}^c w_i(y_{(i-1)p+1}+y_{(i-1)p+2}+\cdots+y_{ip})\equiv 0\pmod 2$$
and so $\pi^{-1}(w)\perp E_{\sigma}(C)$. Hence $\pi^{-1}(w)\perp C$, which proves that $\pi^{-1}(w)\in C^\perp$ and $w\in \pi(F_{\sigma}(C^\perp))=\widehat{C_\pi}$. It follows that $C_{\pi}^{\perp} \subseteq \widehat{C_\pi}$, and therefore $C_{\pi}^{\perp} = \widehat{C_\pi}$.

 Now consider the codes $C_{\varphi}$ and $\widehat{C_{\varphi}}$. If $a(x),b(x)\in \mathcal{P}$ then
 \begin{align*}
   a(x)b(x^{-1}) & =(a_0+a_1x+\cdots+a_{p-1}x^{p-1})(b_0+b_1x^{p-1}+\cdots+b_{p-1}x) \\
    & = a\cdot b + (a\cdot\sigma(b))x + \cdots +
(a\cdot\sigma^{p-1}(b))x^{p-1},
 \end{align*}
  where $a\cdot b=\sum_{i=0}^{p-1}
a_ib_i$ is the inner product of the binary vectors $a$ and $b$.
If $v=(v_1(x),\ldots,v_c(x))$ and $w=(w_1(x),\ldots,w_c(x))$ are vectors in $\mathcal{P}^c$, then
$$
  \sum_{i=1}^c v_i(x)w_i(x^{-1})  =\sum_{i=1}^c v_i\cdot w_i +
(\sum_{i=1}^c v_i\cdot\sigma(w_i))x + \cdots \\
   + (\sum_{i=1}^c v_i\cdot\sigma^{p-1}(w_i))x^{p-1}.
$$

   Let $v\in C_\varphi$, $w\in \widehat{C_\varphi}$, $v'=\varphi^{-1}(v)\in E^*_\sigma(C)$, and $w'=\varphi^{-1}(w)\in E^*_\sigma(C^\perp)$. Then $w'$, $\sigma(w'),\ldots, \sigma^{p-1}(w')\in E^*_\sigma(C^\perp)$, and so
   $$v'\cdot w'=v'\cdot\sigma(w')=\cdots=v'\cdot \sigma^{p-1}(w')=0.$$
   If $v'=(v_{11},\dots,v_{1p},\dots,v_{c1},\dots,v_{cp})$, $w'=(w_{11},\dots,w_{1p},\dots,w_{c1},\dots,w_{cp})$, $v_i=(v_{i1},\dots,v_{ip})$, and $w_i=(w_{i1},\dots,w_{ip})$ then
\begin{eqnarray*}
  v'\cdot w' &=& \sum_{i=1}^c v_i\cdot w_i=0, \\
  v'\cdot \sigma(w') &=& \sum_{i=1}^c v_i\cdot\sigma(w_i)=0, \\
  \vdots &=& \vdots \\
  v'\cdot \sigma^{p-1}(w') &=& \sum_{i=1}^c v_i\cdot\sigma^{p-1}(w_i)=0.
\end{eqnarray*}
   It turns out that $$\langle v,w\rangle=
  \sum_{i=1}^c v_i(x)w_i(x^{-1})  =\sum_{i=1}^c v_i\cdot w_i +
(\sum_{i=1}^c v_i\cdot\sigma(w_i))x + \cdots \\
   + (\sum_{i=1}^c v_i\cdot\sigma^{p-1}(w_i))x^{p-1}=0.
$$
This proves that $w\in C^{\perp}_\varphi$ and so $\widehat{C_\varphi}\subseteq C_{\varphi}^{\perp}$.

If $w\in C_{\varphi}^{\perp}$ then
$$\sum_{i=1}^c w_i(x)v_i(x^{-1}) =\sum_{i=1}^c w_i\cdot v_i +
(\sum_{i=1}^c w_i\cdot\sigma(v_i))x + \cdots + (\sum_{i=1}^c w_i\cdot\sigma^{p-1}(v_i))x^{p-1}=0.$$
It follows that
$$\sum_{i=1}^c w_i\cdot v_i=\sum_{i=1}^c
w_i\cdot\sigma(v_i)=\cdots=\sum_{i=1}^c
w_i\cdot\sigma^{p-1}(v_i)=0.$$
Therefore, if $w'=(w_1,\ldots,w_c,\underbrace{0,\ldots,0}_f)$ then $w'\perp E_\sigma(C)$ and therefore $w'\in C^\perp$. It  turns out that $w'\perp E_\sigma(C^\perp)$ and $w\in \widehat{C_{\varphi}}$.
This gives us that $C_{\varphi}^{\perp} \subseteq \widehat{C_\varphi}$, hence $C_{\varphi}^{\perp} = \widehat{C_\varphi}$.
\end{proof}

Let $\widehat{C_\varphi}=\widehat{M_1}\oplus\cdots\oplus \widehat{M_r}$, $\dim_{I_j} \widehat{M_j}=\widehat{k_j}$, $j=1,\ldots,r$. Then for the dimensions of the involved codes we have
$$\dim C_{\pi}^{\perp}=c+f-k_{\pi}, \ n-k=cp+f-k=c+f-k_{\pi}+m(\widehat{k_1}+\cdots+\widehat{k_r})$$
$$\Rightarrow c(p-1)-k+k_{\pi}=m(\widehat{k_1}+\cdots+\widehat{k_r})$$
$$\Rightarrow c(p-1)-m(k_1+\cdots+k_r)=m(\widehat{k_1}+\cdots+\widehat{k_r})$$
$$\Rightarrow c\frac{p-1}{m}=k_1+\cdots+k_r+\widehat{k_1}+\cdots+\widehat{k_r}$$

For the following assertions, it is more convenient for us to denote the codes $M_{s+j}$ by $M'_j$ and $M_{s+t+j}$ by $M''_j$ for $j=1,\ldots,t$, respectively.  

The considered inner product in $\mathcal{P}^c$ defines an inner product in the linear space $G_i^c$ over the field $G_i$, $i=1,\ldots,s$, but it's not the same for the linear spaces $H_j^c$ and $(H_j^*)^c$. Therefore, for the code $M'_j$, $j=1,\ldots,t$, by $(M'_j)^\perp$ we denote the following linear code over the field $H_j^*$: $(M'_j)^\perp=\{w\in (H_j^*)^c, \langle v,w \rangle=0 \ \forall v\in M'_j\}$. Similarly,
$(M''_j)^\perp=\{w\in H_j^c, \langle v,w \rangle=0 \ \forall v\in M''_j\}$.

\begin{lemma}
The following holds for the subcodes of $C_\varphi$ and $\widehat{C_\varphi}$: (i) $M_i^\perp=\widehat{M_i}$ for $i=1,\ldots,s$; (ii) $(M'_j)^\perp=\widehat{M''_j}$ for $j=1,\ldots,t$; and (iii) $(M''_j)^\perp=\widehat{M'_j}$ for $j=1,\ldots,t$.
\end{lemma}

\begin{proof} Recall that $I_i=G_i$ for $i=1,\ldots,s$, $H_j=I_{s+j}$, $H_j^*=I_{s+t+j}$, $M'_j=M_{s+j}$ and $M''_j=M_{s+t+j}$ for $j=1,\ldots,t$. Moreover, if $e_j$ is the nonzero idempotent of $I_j$ then $e_i(x)e_j(x)=0$ for $i\neq j$, $1\leq i,j\leq r=s+2t$.

\begin{itemize}
\item[(i)] First consider the codes $M_i$ over the fields $G_i$ for $i=1,\dots,s$.
Let $e_i(x)$ be the nonzero idempotent of $G_i$ and $v\in M_i$, $1\le i\le s$. Hence $v=(v_1(x)e_i(x),\ldots,v_c(x)e_i(x))$, $v_1,\ldots,v_c\in\mathcal{R}$, and $e_i(x^{-1})=e_i(x)$. If $w=(w_1,\ldots,w_c)\in M_i^\perp$ and $u=(u_1,\ldots,u_c)\in M_j$ for $j\neq i$, $1\le j\le r$, then
$w=(w_1e_i,\ldots,w_ce_i)$, $u=(u_1e_j,\ldots,u_ce_j)$, and
$$\langle v,w\rangle=(v_1(x)w_1(x^{-1})+\ldots+v_c(x)w_c(x^{-1}))e_i(x)=0,$$
$$\langle u,w\rangle=(u_1(x)w_1(x^{-1})e_j(x)e_i(x)+\ldots+u_c(x)w_c(x^{-1})e_j(x)e_i(x)=0.$$
It turns out that $w\perp C_{\varphi}$ and so $\varphi^{-1}(w)\in C^\perp$. Hence $w\in \widehat{M_i}$.

Now take a vector $w\in \widehat{M_i}$. Hence $w=(w_1(x)e_i(x),\ldots,w_c(x)e_i(x))$ and $\varphi^{-1}(w)\perp C$. It follows that
$\langle w,v\rangle=0$ for all $v\in M_i$ and therefore $w\in M_i^\perp$. This proves that $M_i^\perp=\widehat{M_i}$.

\item[(ii)] Let $e'_j(x)=e_{s+j}(x)$ be the nonzero idempotent of $H_j$ and $v\in M'_j$, $1\le j\le t$. Hence $v=(v_1(x)e'_j(x),\ldots,v_c(x)e'_j(x))$, $v_1,\ldots,v_c\in\mathcal{R}$, and $e'_j(x^{-1})=e''_j(x)$, where $e''_j=e_{s+t+j}$. If $w=(w_1e''_j,\ldots,w_ce''_j)\in (M'_j)^\perp$ and $u=(u_1e_i,\ldots,u_ce_i)\in M_i$ for $i\neq s+j$, $1\le i\le r$, then
$$\langle v,w\rangle=(v_1(x)w_1(x^{-1})+\ldots+v_c(x)w_c(x^{-1}))e'_j(x)=0,$$
$$\langle u,w\rangle=(u_1(x)w_1(x^{-1})e_i(x)e'_j(x)+\ldots+w_c(x)u_c(x^{-1})e_i(x)e'_j(x)=0.$$
It turns out that $w\perp C_{\varphi}$ and so $\varphi^{-1}(w)\in C^\perp$. Hence $w\in \widehat{M''_j}$.

If we take $w\in \widehat{M''_j}$ then $\varphi^{-1}(w)\in C^\perp$ and therefore $\langle w,v\rangle=0$ for all $v\in M'_j$. It follows that $w\in (M'_j)^\perp$. This proves that $(M'_j)^\perp=\widehat{M''_j}$.

\item[(iii)] The proof for  $(M''_j)^\perp=\widehat{M'_j}$ is the same as in the previous case.
\end{itemize}

In this way we proved that

 $~~~~~\widehat{C_{\varphi}}=M_2^\perp\oplus\cdots\oplus M_s^\perp\oplus (M'_1)^\perp\oplus\cdots\oplus(M'_t)^\perp \oplus (M''_1)^\perp\oplus\cdots\oplus(M''_t)^\perp$.
\end{proof}

At the end of this section, we present a theorem that is important in classifying binary linear codes having the automorphism $\sigma$ of type \eqref{perm} in order to reject some of the considered codes up to equivalence.

\begin{theorem}\label{thm:eq}{\rm \cite{Yorgov56}}
The following transformations preserve the decomposition and send
the code $C$  to an equivalent one:

a) substitution $x\to x^t$ in $C_{\varphi}$, where $t$ is an
integer, $1\le t\le p - 1$;

b) multiplication of the $j$th coordinate of $C_{\varphi}$ by
$x^{t_j}$ where $t_j$ is an integer, $0\le t_j\le p - 1$, $j =
1,2,\dots,c$;

c) permutation of the first $c$ cycles of $C$;

d) permutation of the last $f$ coordinates of $C$.
   \end{theorem}

\section{LCD codes}
\label{sect:LCD}

In this section, we study the properties of binary LCD codes invariant under a permutation of odd prime order $p$. Let $C$ be a binary linear code and $\sigma\in\mathrm{Aut}(C)$, where $\sigma$ is the permutation given in \eqref{perm}. We use the structure of $C$, presented in Section \ref{sect:linear}.

\begin{theorem} $C$ is an LCD code if and only if $E_\sigma(C)$ and $F_\sigma(C)$ are LCD codes.
\end{theorem}

\begin{proof} We will prove that if $C=C_1\oplus C_2$ and $C_1\perp C_2$, then $C$ is an LCD code if and only if both $C_1$ and $C_2$ are LCD codes.

$\Rightarrow)$ Let $C$ be an LCD code. If $w=(w_1,\ldots,w_{n})\in C_1\cap C_1^\perp$ then $w\perp C_1$ and $w\perp C_2$. This gives us that $w\perp C$ and so $w\in C\cap C^\perp$. Hence $w=0$ and $C_1$ is an LCD code. The same holds for the code $C_2$.

$\Leftarrow)$ Let $C_1$ and $C_2$ be LCD codes, and $x\in C\cap C^\perp$. Since $C= C_1\oplus C_2$ then $x=x_1+x_2$, $x_i\in C_i$, $i=1,2$. Take $y_i\in C_i$, $i=1,2$. Then we have $x\cdot y_i=0$ and
$$x_i\cdot y_i=(x_1+x_2)\cdot y_i=x\cdot y_i=0 \ \Rightarrow x_i\perp C_i \ \Rightarrow x_i\in C_i\cap C_i^\perp \ \Rightarrow x_i=0, \
i=1,2.$$
This proves that $x=0$ and so $C$ is also an LCD code.

To complete the proof, we take $C_1=E_\sigma(C)$ and $C_2=F_\sigma(C)$.
\end{proof}


\begin{remark}\rm
 Obviuosly, $E_{\sigma}(C)$ is an LCD code if and only if the code $E_{\sigma}(C)^*$ is LCD.
\end{remark}

\begin{lemma} The code $F_\sigma(C)$ is LCD if and only if its image $C_\pi$ is also an LCD code.
\end{lemma}

\begin{proof}
The proof follows immediately from Lemma \ref{lemma:Cpi}. The equality $\pi(F_\sigma(C)\cap F_\sigma(C)^\perp)=C_{\pi}\cap C_{\pi}^{\perp}$ gives us that $F_\sigma(C)\cap F_\sigma(C)^\perp=\{\mathbf{0}\}$ if and only if $C_{\pi}\cap C_{\pi}^{\perp}=\{\mathbf{0}\}$.
\end{proof}

\begin{lemma}\label{thm:phi} The code $E_\sigma(C)^*$ is LCD if and only if its image $C_\varphi$ is also an LCD code.
\end{lemma}

\begin{proof}
First we will prove that $\varphi(E_{\sigma}(C)^*\cap E_{\sigma}(C^{\perp})^*)=C_{\varphi}\cap \widehat{C_{\varphi}}$.

If $v\in E_{\sigma}(C)^*\cap E_{\sigma}(C^{\perp})^*$ then $$\varphi(v)\in\varphi(E_{\sigma}(C)^*)\cap\varphi(E_{\sigma}(C^{\perp})^*)=C_{\varphi}\cap \widehat{C_{\varphi}}.$$
Hence $\varphi(E_{\sigma}(C)^*\cap E_{\sigma}(C^{\perp})^*)\subseteq C_{\varphi}\cap \widehat{C_{\varphi}}$.

Consequently, $E_{\sigma}(C)^*\cap E_{\sigma}(C^{\perp})^*=\{\mathbf{0}\}$ if and only if $C_{\varphi}\cap \widehat{C_{\varphi}}=\{\mathbf{0}\}$. This result proves the lemma.
\end{proof}

We summarize the results from the above lemmas in the following theorem.

\begin{theorem}\label{thm:main}
The binary code $C$ having an automorphism $\sigma$ of odd prime order $p$ is an LCD code if and only if $C_\varphi$  and $C_\pi$ are LCD codes.
\end{theorem}

The following theorem gives more detailed conditions for the code $C_{\varphi}$ to be LCD under the inner product \eqref{eq:product}.

\begin{theorem}\label{thm:Cphi} The code $C_{\varphi}$ is an LCD code under the inner product \eqref{eq:product} if and only if $M_i\cap M_i^{\perp}=\{0\}$,  $i=1,\ldots,s$, $M'_i\cap (M''_i)^{\perp}=\{0\}$ and $M''_i\cap (M'_i)^{\perp}=\{0\}$, $j=1,\ldots,t$.
\end{theorem}
\begin{proof}
If $C_{\varphi}$ is an LCD code and $v\in M_i\cap M_i^{\perp}$, $1\le i\le s$, then $v\in \widehat{M_i}$ and so $v\in \widehat{C_\varphi}=C_\varphi^\perp$. Hence $v=0$ and $M_i$ is an LCD code over the field $G_i$ under the inner product \eqref{eq:product}. If $v\in M'_j\cap (M''_j)^\perp$, $1\le j\le t$, then $v\in \widehat{M'_j}$ and so $v\in \widehat{C_\varphi}=C_\varphi^\perp$. The same follows when $v\in M''_j\cap (M'_j)^\perp$. It turns out that $M'_i\cap (M''_i)^{\perp}=\{0\}$ and $M''_i\cap (M'_i)^{\perp}=\{0\}$, $j=1,\ldots,t$.

In the other hand, if $M_i\cap M_i^{\perp}=\{0\}$ for $i=1,\ldots,s$, $M'_i\cap (M''_i)^{\perp}=\{0\}$, $M''_i\cap (M'_i)^{\perp}=\{0\}$, for $j=1,\ldots,t$, and $v\in C_{\varphi}\cap C_{\varphi}^\perp=C_{\varphi}\cap \widehat{C_{\varphi}}$, then
$$v=v_1+\cdots+v_r=\widehat{v_1}+\cdots+\widehat{v_r},$$
where $v_i\in M_i$, $\widehat{v_i}\in \widehat{M_i}$, $i=1,\ldots,r$. Since the sum is direct, we have $v_i=\widehat{v_i}$ for all $i=1,\ldots,r$. Hence $v_i\in M_i\cap \widehat{M_i}$. It follows that $v_2=\ldots=v_s=0$ and $v_j\in M'_j\cap \widehat{M'_j}=M'_j\cap (M''_j)^\perp$ for $1\le j\le t$, so $v_{s+1}=\ldots=v_r=0$. It turns out that $v=0$ and $C_{\varphi}$ is an LCD code under the inner product \eqref{eq:product}.
\end{proof}

For the primes $p<30$, 2 is a multiplicative root modulo $p$ for $p=3$, 5, 11, 13, 19, and 29. For these values of $p$, the code $C_{\varphi}$ is a linear code over the field $\mathcal{P}\cong \F_{2^{p-1}}$.

If $p=7$, 17, and 23, the multiplicative order of 2 modulo $p$ is equal to $\frac{p-1}{2}$. Moreover,
\begin{align*}
  x^7-1 & = (x-1)(x^3+x+1)(x^3+x^2+1)=(x-1)h(x)h^*(x),\\
  x^{17}-1 & = (x - 1) (x^8 + x^5 + x^4 + x^3 + 1) (x^8 + x^7 + x^6 + x^4 + x^2 + x + 1)=(x-1)g_1(x)g_2(x),\\
  x^{23}-1 &=(x - 1) (x^{11} + x^9 + x^7 + x^6 + x^5 + x + 1) (x^{11} + x^{10} + x^6 + x^5 + x^4 + x^2 + 1)\\
  & =(x-1)h_1(x)h_1^*(x).
\end{align*}
The case $p=17$ is different from the other two, because the irreducible factors $g_1(x)$ and $g_2(x)$ are self-reciprocal polynomials.

\section{Code Constructions}
\label{sect:aut7}

In this section, we present some applications of the presented theory. We give construction and classification results for LCD codes having automorphisms of prime order $p$ for $p=7$, 11, and 17. Many of the constructed LCD codes are optimal as linear codes \cite{Grassl-table}. All calculations related to equivalences, automorphism groups and weight enumerators of the considered binary codes were executed with the software package \textsc{Q-Extension} \cite{Q-Extension}.

\subsection{$p=7$}

Let $C$ be an LCD $[n,k,d]$ code having an automorphism of order 7 with $c=4$ independent 7-cycles and $f$ fixed points. Then $C_{\pi}$ is a binary LCD $[4+f,k_{\pi}]$ code, and $C_{\varphi}=M_1\oplus M_2$, where $M_j$ is a linear code of length $4$ over the field $I_j\cong\F_8$, $j=1,2$. Denote $\dim_{I_j}M_j$ by $k_j$, $j=1,2$. Since
$$x^6+x^5+x^4+x^3+x^2+x+1=(x^3+x+1)(x^3+x^2+1),$$
we can take $I_j=\{0,e_j(x),xe_j(x),\ldots,x^6e_j(x)\}$, $j=1,2$, where $e_1(x)=1+x+x^2+x^4$ and $e_2(x)=1+x^3+x^5+x^6$ are the corresponding idempotents. Since $e_1(x^{-1})=e_2(x)$ and $e_1(x)e_2(x)=0$, we have $e_1(x)e_1(x^{-1})=e_1(x)e_2(x)=0$ and so the code $M_1$ is self-orthogonal under the defined inner product (the same for $M_2$). The presented structure gives us that the binary codes $E_1=\varphi^{-1}(M_1)$ and $E_2=\varphi^{-1}(M_2)$ are doubly-even $[7c,3k_1,4d_1]$ and $[7c,3k_2,4d_2]$ codes, respectively, and $E_\sigma(C)^*=E_1\oplus E_2$.  As the substitution $x\to x^3$ in $C_\varphi$ interchanges $M_1$ and $M_2$, we may assume that $k_1\ge k_2$. To obtain more restrictions on $k_1$ and $k_2$, we use the following lemma.

\begin{lemma}\label{lem:hull} If $C$ is a binary linear code of length $n$ with $\dim(C\cap C^\perp)=s\ge 1$, and $C_1=\langle C,x\rangle$ for $x\in\F_2^n$, $x\not\in C$, then $\dim(C_1\cap C_1^\perp)\ge s-1$.
\end{lemma}

\begin{proof} The intersection $\mathcal{H}(C)=C\cap C^\perp$ is called the hull of the code $C$. Let $s=\dim\mathcal{H}(C)$. If $\mathcal{H}_0=\{ v\in \mathcal{H}(C) \ \vert \ v\cdot x=0\}$ then $\mathcal{H}_0$ is a subcode of $\mathcal{H}(C)$ of dimension $s$ or $s-1$. Moreover,  $\mathcal{H}_0\subseteq C\subset C_1$ and $\mathcal{H}_0\perp C_1$, so $\mathcal{H}_0\subseteq \mathcal{H}(C_1)=C_1\cap C_1^\perp$. It follows that $\dim(C_1\cap C_1^\perp)\ge\dim \mathcal{H}_0\ge s-1$.
\end{proof}

\begin{corollary}
Let $C$ be a binary LCD code having an automorphism of order 7 with $c$ independent 7-cycles. If $\dim_{I_1}M_1=k_1$ and $\dim_{I_2}M_2=k_2$, then $k_1=k_2$.
\end{corollary}

\begin{proof}
As we have already mentioned, we can take $k_1\ge k_2$. Then $E_1=\varphi^{-1}(M_1)$ is a binary $[7c,3k_1]$ doubly-even code and therefore $\dim \mathcal{H}(E_1)=3k_1$. Applying Lemma \ref{lem:hull} $3k_2$ times, we have that $\dim \mathcal{H}(E_\sigma(C)^*)\ge 3k_1-3k_2$. Since $C$ is an LCD code, $E_\sigma(C)^*$ is also a binary LCD code and therefore $\dim \mathcal{H}(E_\sigma(C)^*)=0$. This proves that $k_1-k_2=0$.
\end{proof}

Let $c=4$ and $\sigma=(1,2,\ldots,7)(8,9,\ldots,14)(15,15,\ldots,21)(22,23,\ldots,28)$. Then $M_j$ is a $[4,k_j,d_j]$ linear code over $I_j\cong \F_8$, $j=1,2$. Consider first the case $k_1=k_2=1$. We are looking for optimal binary LCD $[28+f,6+k_\pi,d]$ codes for different values of $f$ and $k_\pi$ such that $d\ge 11$.

In this case $E_\sigma(C)^*$ is a binary LCD $[28,6]$ code. We can take $M_1=\langle(e_1,e_1,e_1,0)\rangle$ or $\langle(e_1,e_1,e_1,e_1)\rangle$, up to equivalence. The  $[4,1,3]$ code $M_1$ provides two inequivalent $[28,6,12]$ binary LCD codes, namely with $M_2=\langle(e_2,xe_2,0,e_2)\rangle$ and $M_2=\langle(e_2,xe_2,x^2e_2,e_2)\rangle$.
The  $[4,1,4]$ code $M_1$ provides only one $[28,6,12]$ binary LCD code, for which $M_2=\langle(e_2,xe_2,x^2e_2,x^5e_2)\rangle$. All three codes have the same weight enumerator $1+21z^{12}+21z^{14}+14z^{16}+7z^{18}$. Their automorphism groups have orders 14, 7, and 42, respectively.

     Now consider codes with an automorphism $\sigma$ with $f\ge 1$ fixed points.

\begin{itemize}
\item $f=1$, $k_\pi=1$) In this case $C_\pi$ is a binary LCD $[5,1,3$ or 5] code. Considering all possibilities for $C_\pi$, combined with the three $[28,6,12]$ codes presented above, we obtain 24 inequivalent $[29,7]$ LCD codes but none of them has minimum distance 12. Actually, we checked all possibilities including the codes $C_\pi$ for which the fixed points are 0's. This gives us that for any binary LCD code $C$ of length $28+f$, such that $\sigma\in\mathrm{Aut}(C)$ and $C_{\pi}$ contains a codeword of weight 3, the minimum weight $d$ is less than 12.

\item $f=2$, $k_\pi=2$) Now $C_\pi$ is a binary LCD $[6,2,\ge 2]$ code. Taking in mind the result in the previous case, we consider the following generator matrices for $C_\pi$:
    \[
    \left(\begin{array}{cc}
    1100&00\\
    1010&00\end{array}\right), \ \
    \left(\begin{array}{cc}
    1100&00\\
    1011&10\end{array}\right), \ \
    \left(\begin{array}{cc}
    1100&00\\
    1010&11\end{array}\right), \ \
    \left(\begin{array}{cc}
    1100&00\\
    1011&11\end{array}\right),
    \]
    and also all equivalent to these codes, obtained by permuting the four cycles. All constructed $[30,8]$ LCD codes have minimum distance $\le 10$. Moreover, if we take only the direct sum of the code $E_\sigma(C)$ and an $[c+f,1,2]$ code $F_\sigma(C)$, we obtain a code with minimum weight at most 10.
\end{itemize}

The presented constructions give the following result.

\begin{proposition}
There are three inequivalent $[28,6,12]$ LCD codes having an automorphism of order 7. No optimal LCD $[28+i,6+i,12]$ code for $i\ge 1$ has an automorphism of order 7.
\end{proposition}

\begin{proof}
According to Table 1 in \cite{optLCD}, binary LCD $[28+i,6+i,12]$ codes exist for $i=0,1,2$, but the existence of $[31,9,12]$, $[32,10,12]$, $[33,11,12]$, and $[34,12,12]$ LCD codes is not known. The largest possible minimum weight of a binary LCD $[28+i,6+i,12]$ code for $i\ge 7$ i less than 12.

If we take $c=1$, 2 or 3, the code $E_\sigma(C)*$ would be a binary $[7,6,\le 2]$, $[14,6\le k\le 12,\le 4]$, and $[21,6\le k\le 18,\le 8]$ code, respectively \cite{Grassl-table}. Hence $c=4$ and $E_{\sigma}(C)^*$ is equivalent to one of the constructed $[28,6,12]$ LCD codes. Take $f=i\le 6$. The cases $f=0$, 1, and 2 are considered above. Take $3\le f\le 6$. The $C_{\pi}$ is a binary $[4+f,f]$ LCD code. The optimal LCD $[7,3]$, $[8,4]$, $[9,5]$, and $[10,6]$ codes have minimum weight 3 \cite{Harada-LCD}, and according to the above computations, the corresponding LCD $[28+i,6+i]$ codes will have minimum weight less than 12.
\end{proof}

\subsection{$p=11$}

Since $\mathrm{ord}_{11}(2)=10$, in this case $\mathcal{P}$ is a field with $2^{10}$ elements. Its identity element is $e=x+x^2+\cdots+x^{10}$, and all nonzero elements can be written in the form $\alpha^i\beta^j\gamma^k$, $0\le i\le 30$, $0\le j\le 2$, $0\le k\le 10$, where $\alpha=x^9+x^2$, $\beta=x^{10}+x^8+x^7+x^6+x^2+1$, $\gamma=xe$, are elements of orders 31, 3, and 11, respectively.

We are looking for binary LCD codes of length $n\ge 33$ and minimum distance $d\ge 12$ having an automorphism $\sigma=(1,2,\ldots,11)(12,13,\ldots,22)(23,24,\ldots,33)$. Then $C_{\varphi}$ is an LCD $[3,k_{\varphi}]$ code over the field $\mathcal{P}$ which is an image of a binary LCD $[33, 10k_{\varphi},\ge 12]$ code. It turns out that $k_{\varphi}=1$ and
$$C_{\varphi}=\langle(e,\alpha^{i_1}\beta^{j_1},\alpha^{i_2}\beta^{j_2})\rangle, \ \ 0\le i_1,i_2\le 30, \ 0\le j_1,j_2\le 2, $$
such that $e+\alpha^{i_1}+\alpha^{i_2}\neq 0$.

Considering all possibilities for $i_1,j_1,i_2,j_2,$ we obtain exactly three inequivalent binary LCD $[33,10,12]$ codes $E_{\sigma}(C)^*$. These are the codes $C_{0,1,5,0}$, $C_{1,1,22,1}$, and $C_{5,0,10,0}$, constructed for the corresponding values of $(i_1,j_1,i_2,j_2)$. The first two codes have the same weight enumerator $1+132y^{12}+187y^{14}+242y^{16}+286y^{18}+110y^{20}+55y^{22}+11y^{24}$, and the weight enumerator of the third code is
$1+99y^{12}+275y^{14}+209y^{16}+198y^{18}+187y^{20}+55y^{22}$.

According to our construction, LCD $[33,11,12]$ codes with an automorphism of order 11 do not exist, but there are two optimal LCD $[35,11,12]$ codes, obtained from $C_{0,1,5,0}$ with $C_{\pi}=\langle (00111)\rangle$ and $C_{5,0,10,0}$ with $C_{\pi}=\langle (11111)\rangle$.

\subsection{$p=17$}

In this case the module $\mathcal{P}$ is a direct sum of two fields with 256 elements each, i.e. $\mathcal{P}=G_1\oplus G_2$ where $G_i=\langle g_i(x)\rangle$, $i=1,2$, $g_1(x)= 1+x+x^3+x^6+x^8+x^9$ and $g_2(x) = 1 + x^3 + x^4 + x^5 + x^6 + x^9$ \cite{Yankov_p17}. The nonzero idempotents of these ideals are $e_1 = x + x^2 + x^4 + x^8 + x^9 + x^{13} + x^{15} + x^{16}$ and $e_2 = x^3 + x^5 + x^6 + x^7 + x^{10} + x^{11} + x^{12} + x^{14}$, respectively.  The element $\delta=g_1(x)^{17}=x^3 +x^7 +x^8 +x^9 +x^{10} +x^{14}\in G_1$ has order 17, and so $G_1=\{0, x^i\delta^j | 0\le i\le 16, 0\le j\le 14\}$. Similarly, we take $\tau=g_2(x)^{17}= x + x^3 + x^8 + x^9 + x^{14} + x^{16}\in G_2$ as an element of order 15, and consider $G_2=\{0, x^i\tau^j | 0\le i\le 16, 0\le j\le 14\}$.

In this case we take $c=2$ and $\sigma=(1,2,\ldots,17)(18,19,\ldots,34)$. Then $C_{\varphi}=M_1^{(17)}\oplus M_2^{(17)}$, where $M_i^{(17)}$ is an LCD code over $G_i$ of length 2.

\begin{itemize}
\item $k_1=1$, $k_2=0$) Now $M_1^{(17)}=\langle(e_1,x^i\delta^j)\rangle$. We are looking for LCD codes $C$ with minimum weight $d\ge 13$. Considering the possibilities for $j$ we obtain that only for $j=5$ and 10 the corresponding code $M_1^{(17)}$ has the needed minimum weight. According to Theorem \ref{thm:eq}, the codes $E_\sigma(C)^*$ for all values of $i$ and for $j=5$ and $10$ are equivalent to one binary LCD $[34,8,14]$ code. This code is optimal, its automorphism group has order 136 and its weight enumerator is $1+68y^{14}+68y^{16}+68y^{18}+34y^{20}+17y^{24}$.

    Further, we use the binary LCD $[f+2,k_{\pi}=1]$ codes $C_{\pi}$. The minimum weights of the cosets $(\underbrace{11\ldots 1}_{17}\underbrace{00\ldots 0}_{17})+E_\sigma(C)^*$, $(\underbrace{00\ldots 0}_{17}\underbrace{11\ldots 1}_{17})+E_\sigma(C)^*$, and $(11\ldots 1)+E_\sigma(C)^*$ are 13, 13, and 10, respectively. We obtain two optimal LCD codes: (1) the $[34,9,13]$ code $C_{34,9}$ with weight enumerator $1+51y^{13}+68y^{14}+68y^{15}+68y^{16}+18y^{17}+68y^{18}+68y^{19}+34y^{20}+51y^{21}+17y^{24}$, and (2) the $[36,9,14]$ code $C_{36,9}$ with weight enumerator $1+68y^{14}+51y^{15}+68y^{16}+68y^{17}+68y^{18}+18y^{19}+34y^{20}+68y^{21}+51y^{23}+17y^{24}$. Both codes have automorphism groups of order 68. We have made some computations for larger values of $k_{\pi}$ and $f$ but the constructed LCD codes are not optimal.

\item $k_1=k_2=1$) We are looking for LCD codes $C$ with minimum weight $d\ge 8$. After computing the minimum distances and checking for equivalence, from all codes $C_{\varphi}=M_1^{(17)}\oplus M_2^{(17)}$ with $M_1^{(17)}=\langle(e_1,\delta^{j_1})\rangle$ and $M_2^{(17)}=\langle(e_2,x^i\delta^{j_2})\rangle$, we obtain exactly 10 $[34,16,8]$ LCD codes with 7 different weight enumerators. These codes are not optimal as linear codes but as far as we know they are the first constructed LCD $[34,16]$ codes with minimum weight $>7$. The values of the parameters $j_1$, $j_2$ and $i$, as well as their weight enumerators, are given in Table \ref{Table_34_16}. Adding a $[2,1,1]$ LCD code $C_{\pi}$, we obtain three optimal $[34,17,8]$ LCD codes. We present some of their characteristics in Table \ref{Table_34_17}.
\end{itemize}

\begin{table}\label{Table_34_16}
\small
\begin{center}
\begin{tabular}{|c|ccc|l|}
  \hline
  &$j_1$&$j_2$&$i$&Weight enumerator\\
  \hline
 1&0&0&1&$1+153y^{8}+952y^{10}+4369y^{12}+\cdots+935y^{24}+136y^{26}+17y^{28}$\\
 2&0&0&2&$1+153y^{8}+952y^{10}+4369y^{12}+\cdots+935y^{24}+136y^{26}+17y^{28}$ \\
 3&0&0&3&$1+170y^{8}+918y^{10}+4318y^{12}+\cdots+884y^{24}+102y^{26}+34y^{28}$ \\
 4&0&0&5&$1+153y^{8}+952y^{10}+4369y^{12}+\cdots+935y^{24}+136y^{26}+17y^{28}$ \\
 5&0&0&7&$1+119y^{8}+1071y^{10}+4284y^{12}+\cdots+901y^{24}+187y^{26}$ \\
 6&3&0&0&$1+187y^{8}+884y^{10}+4267y^{12}+\cdots+833y^{24}+68y^{26}+51y^{28}$\\
 7&3&0&2&$1+204y^{8}+918y^{10}+4148y^{12}+\cdots+1054y^{24}+102y^{26}$\\
 8&3&0&3&$1+153y^{8}+901y^{10}+4556y^{12}+\cdots+1071y^{24}+153y^{26}$ \\
 9&3&0&6&$1+204y^{8}+918y^{10}+4148y^{12}+\cdots+1054y^{24}+102y^{26} $\\
 10& 5&0&0&$1+102y^{8}+1156y^{10}+4148yz^{12}+\cdots+816y^{24}+204y^{26}$\\
  \hline
\end{tabular}
\vspace{0.3cm} \\
\end{center}
\caption{LCD $[34,16,8]$ codes}
\end{table}

\begin{table}\label{Table_34_17}
\small
\begin{center}
\begin{tabular}{|c|ccc|c|l|}
  \hline
  &$j_1$&$j_2$&$i$&$C_{\pi}$&Weight enumerator\\
  \hline
 1&0&0&3&(01)&$1+170y^{8}+527y^{9}+918y^{10}+\cdots+391y^{25}+102y^{26}+34y^{27}+34y^{28}$\\
 2&3&0&0&(10)&$1+187y^{8}+493y^{9}+884y^{10}+\cdots+493y^{25}+68y^{26}+51y^{28}$ \\
 3&3&0&6&(01)&$1+204y^{8}+442y^{9}+918y^{10}+\cdots+374y^{25}+102y^{26}+34y^{27}+17y^{29} $ \\
   \hline
\end{tabular}
\vspace{0.3cm} \\
\end{center}
\caption{LCD $[34,17,8]$ codes}
\end{table}

The results in this section show that the presented method is a powerful tool for constructing optimal binary LCD codes with different parameters.

\section*{Acknowledgements}

The research of S. Bouyuklieva was supported by a Bulgarian NSF contract KP-06-N32/2-2019.

\end{document}